%% file: local_coupling_stalling_FDR-arXiv-submission.tex
\documentclass[aps,prl,reprint,floatfix]{revtex4-1}
\usepackage{amsmath,amssymb}
\usepackage{amsthm}
\usepackage{graphicx}
\usepackage{color}
\usepackage[utf8]{inputenc}

\input{responsemacros}

\newtheorem{theorem}{Theorem}
\newtheorem{corr}{Corollary}

\begin{document}
\title{Fluctuation-dissipation relations far from equilibrium} 
\author{Bernhard Altaner}
\author{Matteo Polettini}
\author{Massimiliano Esposito}
\affiliation{Complex Systems and Statistical Mechanics, Physics and Materials Science Research Unit, University of Luxembourg, Luxembourg}

\begin{abstract}
  Near equilibrium, where all currents of a system vanish on average, the fluctuation-dissipation relation (FDR) connects a current's spontaneous fluctuations with its response to perturbations of the conjugate thermodynamic force.
  Out of equilibrium, fluctuation-response relations generally involve additional nondissipative contributions.
  Here, in the framework of stochastic thermodynamics, we show that an equilibrium-like FDR holds for internally equilibrated currents, if the perturbing conjugate force only affects the microscopic transitions that contribute to the current.
  We discuss the physical requirements for the validity of our result and apply it to nano-sized electronic devices.
  \end{abstract}

  \maketitle

  According to statistical mechanics, systems at equilibrium enjoy a special property:
  it is impossible to tell their spontaneous fluctuations from their response to small external perturbations.
  This message lies at the heart of so-called fluctuation-dissipation relations (FDR)~\cite{Callen.Welton1951,Green1954,Kubo1957,Marconi.etal2008}.
  However, most complex systems live out of equilibrium.
  Equilibrium-like conditions can only be reproduced artificially in localized patches, whereby some particular current stalls in the presence of other currents which are sustained by nonequilibrium driving forces.
  Then, does the FDR hold for such stalled currents?
  Or, in reverse: Is the validity of the FDR a genuine hallmark for equilibrium systems?

  Nonequilibrium statistical mechanics provides the modern toolbox to tackle such questions.
  In particular, the framework of stochastic thermodynamics gives a thermodynamic description of small systems subject to fluctuations~\cite{Sekimoto1998,Jarzynski2011,Seifert2012,VandenBroeck.Esposito2015,Qian2005}, with applications to interdisciplinary areas including nanoscopic electronics \cite{Strasberg.etal2013,Cuetara.Esposito2015,Fujisawa.etal2006,Koski.etal2015,Sanchez.etal2013}, complex bio-molecules such as molecular motors~\cite{Liepelt.Lipowsky2007,Lau.etal2007,Gaspard.Gerritsma2007,Seifert2011,Golubeva.Imparato2012,Altaner.etal2015}, and chemical reaction networks \cite{Gaspard2004,Schmiedl.Seifert2007,Ge.etal2012,Polettini.etal2015a}.
  In this framework, the stochastic observables of experimental interest are the time-averaged thermodynamic currents $\frac{1}{\tau}  \Phi_\alpha^{(\tau)}$, where $\Phi_\alpha^{\tau} := \int \jmath_\alpha(t) \df t$ is the time-integral over an instantaneous fluctuating current $\jmath_\alpha(t)$ (\eg of matter, heat, charge etc.).
  Due to the limited accuracy of measurements, often only the first two cumulants of a current's steady-state statistics are accessible:
  The expected behavior is expressed by the average value $J_\alpha$, whereas fluctuations are characterized by a generalized diffusion constant $D_{\alpha,\alpha}$, obtained from the scaling of the generalized mean square displacement $\langle \Phi_\alpha^{(\tau)}  \Phi_\alpha^{(\tau)}\rangle$ with time $\tau$.
  Above and in what follows, $\langle\, \cdot\, \rangle$ denotes an average over stochastic trajectories sampled from a stationary ensemble.
  According to phenomenological thermodynamics, the steady state dissipation rate \mbox{$\kb \sum_\alpha J_\alpha h_\alpha$}, is a bi-linear form of physical currents $J_\alpha$ and their conjugate thermodynamic forces $h_\alpha$ (\eg gradients in chemical potential or temperature, electrical fields \textit{etc.}).
  In the following, we work with dimensionless quantities and set Boltzmann's constant $\kb$ to unity.
  The FDR connects the dissipative response of one current $J_\alpha$ (\ie the response with respect to a variation of its conjugate force $h_\alpha$) with its spontaneous fluctuations:
  \begin{align}
    \partial_{h_\alpha} J_\alpha (\xs^{\text{eq}})=\diffu_{\alpha,\alpha}(\xs^{\text{eq}}).
    \label{eq:equilibriumFDR}
  \end{align}
  Notice that we introduced an explicit dependence on a vector $\xs$ of arbitrary parameters characterizing the system and its environment.
  In what follows, we consider the force $h_\alpha$ as an independent parameter, such that $\xs = \xs(h_\alpha)$. 
  The FDR~\eqref{eq:equilibriumFDR} requires that the physical parameters identify an equilibrium system, \ie at $\xs = \xs^{\text{eq}}$ all thermodynamic forces and thus all currents
  vanish.

  The question whether and how a result analogous to Eq.~\eqref{eq:equilibriumFDR} can be extended to nonequilibrium situations has attracted considerable attention~\cite{Blickle.etal2007,Chetrite.etal2008,Andrieux.Gaspard2007a,Baiesi.etal2009,Gomez-Solano.etal2009,Prost.etal2009,Seifert2010,Seifert.Speck2010,Verley.etal2011}. 
  The general understanding is that the FDR has to be modified in nonequilibrium situations by considering additional correlations with a time-symmetric quantity, often called activity~\cite{Baiesi.etal2009,Baerts.etal2013}.
  Hence, the usual nonequilibrium extensions of fluctuation-dissipation relations are rather fluctuation-dissipation-activity relations than true FDRs.
  Moreover, most of the above cited results are either formal or formulated in the context of specific setups.
  To our best knowledge \footnote{
    Like in the present case, the abstract formula stated in Ref.~\cite{Seifert2010} considers the response of currents to thermodynamically conjugate forces.
    However, we have found there is a mistake in the proof of this paper and identified cases where this formula gives wrong results.
  },
  a physical picture has only been obtained for conservative perturbations, where stochastic transition rates between two states are modified anti-symmetrically by the addition of a potential $V$, rather than changing a nonconservative driving force~\cite{Baiesi.Maes2013}.%

  In this Letter, for the first time, we present clear conditions for the validity of a true FDR in situations far from equilibrium.
  Our main result states that given a force $h_\alpha$, which couples only to those transitions that contribute to the conjugate current $J_\alpha$, we find a nonequilibrium fluctuation-response relation which takes the equilibrium form:
  \begin{align}
    \partial_{h_\alpha} J_\alpha (\xs^{\star})=\diffu_{\alpha,\alpha}(\xs^{\star}).
    \label{eq:stallingFDR}
  \end{align}
  Crucially, the validity of Eq.~\eqref{eq:stallingFDR} requires that we consider the response and fluctuations at parameter values $\xs^\star$, where the current $J_\alpha(\xs^\star) =0$ stalls \emph{internally}:
  all contributing stochastic transitions need to be internally equilibrated, \ie they are microscopically reversible.

  \paragraph{Setup ---}
  We consider a generic system with a finite number of states $n \in \{ 1,2,\dots,N \}$.
  Possible transitions between states form a connected network, where we draw one edge $\edge$ connecting two states  for each distinguishable physical mechanism by which the transition may occur.
  Stochastic thermodynamics requires that transitions along an edge $\edge$ are always possible in both directions~\cite{Esposito2012,Seifert2011}.
  In some cases it is thus useful to consider $\edge=(+\edge,-\edge)$ as a pair of directed edges $\pm\edge$.
  The evolution of the system is modeled as a Markov jump process and can be visualized as a random walk on the network. 
  A physical model is defined by prescribing the forward and backward transition rates $w_{+\edge}(\xs)$ and  $w_{-\edge}(\xs)$ for each edge $\edge$ as functions of a set of physical parameters $\xs$.
  The fluctuating current along an edge~$\edge$, $\jmath_\edge{(t)} := \sum_{k}\delta(t-t_j)(\delta_{+\edge,\edge_k} - \delta_{-\edge,\edge_k})$, is a stochastic variable, which peaks if the system transitions along the directed edge $\edge_k$ at a jump time $t_k$.
  Physical fluctuating currents $\jmath_\alpha$ which are associated to the transport of a physical quantity (particles, energy, \etc) are weighted edge currents, $\jmath_\alpha = \sum_\edge \inc^\alpha_\edge \jmath_\edge$,
  where $\inc^\alpha_{+\edge} = -\inc^\alpha_{-\edge}$ specifies the amount $\inc^\alpha_\edge$ exchanged with external reservoirs upon a transition along edge $\edge$~\cite{Seifert2011}.
  Ergodicity ensures that time-integrated currents $\Phi^{(\tau)}_\alpha$ are almost surely extensive in time, \ie for $\tau \to \infty$, $\tau^{-1}\Phi^{(\tau)}_\alpha\to J_\alpha$.
  Deviations from the stationary average $J_\alpha$ scale diffusively, and a generalized diffusion constant $D_{\alpha,\beta}$ (or, equivalently $D_{\edge,\edge'}$) is obtained as a correlation integral \cite{Lebowitz.Spohn1999,Altaner.etal2015}:
  \begin{align}
    \diffu_{\alpha\beta} :&= \lim_{\tau\to\infty} \frac{1}{2\tau} \langle \Phi^{(\tau)}_\alpha\Phi^{(\tau)}_\beta\rangle\nonumber \\ &= \int_0^\infty \langle (\jmath_\alpha(0) - J_\alpha)\,(\jmath_\beta(t) - J_\beta)\rangle.
    \label{eq:diffusivity}
  \end{align}


  \paragraph{Local detailed balance, conjugacy and local coupling ---}

  A central assumption in stochastic thermodynamics is \textit{local detailed balance} (LDB), which relates physical currents $J_\alpha$ to their conjugate forces $h_\alpha$.
  The latter are uniquely determined by the intensive parameters of the reservoirs.
  Besides justifying Markovian dynamics, LDB also ensures thermodynamic consistency~\cite{Seifert2011,Esposito2012,Polettini.etal2016}.
  It enters as a constraint on the motance $\mot_\edge$~\cite{Altaner.etal2012}, defined as the log-ratio of forward and backward transition rates:
  \begin{align}
    \mot_\edge := \ln\left(\frac{\w_{+\edge}}{\w_{-\edge}}\right) \stackrel{\text{(LDB)}}{=} \sum_\alpha h_\alpha \inc^\alpha_\edge.
    \label{eq:LDB}
  \end{align}
  Physically, the motance characterizes the entropy change in the system's environment associated to the transition along edge $\edge$.
  Thus, for systems obeying LDB, the stochastic steady-state entropy production rate $\sum_\edge J_\edge \mot_\edge$ takes the well-known bi-linear form involving currents and forces $ \sum_\alpha J_\alpha h_\alpha$.
  In the following, we consider a single physical current $\jmath_\alpha$ supported on a subset of edges  $\edges_\alpha = \left\{\edge: d^\alpha_\edge \neq 0\right\}$.
  Even without the full knowledge about other physical currents $j_\beta$ and forces $h_\beta$, a definition of the current which is conjugate to the force $h_\alpha$ is possible.
  In accordance with LDB, Eq.~\eqref{eq:LDB}, the conjugate force $h_\alpha$ is a parameter that contributes linearly to the motances $\mot_\edge$,  with a slope determined by the edge increments: $\inc^\alpha_\edge = \partial_{h_\alpha}\mot_\edge  $.
  Then, it can be shown~\cite{Polettini.etal2016} that a term $h_\alpha J_\alpha$ appears as in independent term in the stochastic entropy production.
  We further say that a force $h_\alpha$ couples \emph{locally} to its conjugate current $\jmath_\alpha$, if $\partial_{h_\alpha}\w_{\pm\edge}=0$ for all noncontributing edges $\edge \not \in \edges_\alpha$.
  While conjugacy ensures that the force $h_\alpha$ does not change the anti-symmetric part of the rate pair $\w_{\pm\edge}$ for noncontributing edges, locality also demands that it does not affect their symmetric (kinetic) properties.

  \paragraph{Stalled currents and internal equilibrium ---}
  Our main result Eq.~\eqref{eq:stallingFDR} concerns systems where a particular current $\jmath_\alpha$ stalls, \ie where its average vanishes: $J_\alpha = 0$.
  While \emph{all} currents stall in systems that are at equilibrium, $\xs = \xs^\text{eq}$, in nonequilibrium systems, a single current $J_\alpha$ may stall, while other currents are of arbitrary magnitude.
  In general one can tune the conjugate force $h_\alpha$ to its stalling value $h_\alpha^\star$, which is a function of the remaining system parameters.
  Considering stalled currents in systems far from equilibrium is not merely a mathematical exercise:
  In many experiments on molecular motors, applying a mechanical force in order to stall the motor velocity is used to infer the force generated at a given value of the chemical concentrations~\cite{Juelicher.Bruinsma1998, Carter.Cross2005}.
  In nanoscopic electronic devices connected to several leads, stalling several currents in the presence of other nonvanishing physical currents is at the heart of what is known as a B\"uttiker probe \cite{Buettiker1986,Dubi.DiVentra2011,Sanchez.etal2013,Brandner.etal2013}.
  

  The validity of our result requires not only phenomenological stalling, but internal stalling.
  For a current that is supported on a single edge only, phenomenological stalling and internal stalling are equivalent.
  For a physical current with contributions from multiple transitions, phenomenological stalling is not sufficient.
  If in a system the vanishing of a physical current also gives rise to internal stalling, then there is a lot more structure to its internal transitions.
  We will discuss a physical model that features internal stalling below.
  Internal stalling further implies that the turmoils intrinsic to the transitions contributing to a physical current do not lead to any internal entropy production.
  To appreciate this fact, note that a stalled current, $J_\alpha = \sum_{\nu=1,2} \inc^\alpha_{\edge_{\nu}}J_{\edge_\nu}  = 0$, only has a vanishing associated entropy production $\sum_{i=1,2} \mot_{\edge_{\nu}}J_{\edge_\nu}$, if the edge motances obey the condition $\inc^\alpha_{\edge_{2}}\mot_{\edge_{1}}=\inc^\alpha_{\edge_{1}}\mot_{\edge_{2}}$.

  \begin{figure}[t]
    \centering
    \includegraphics{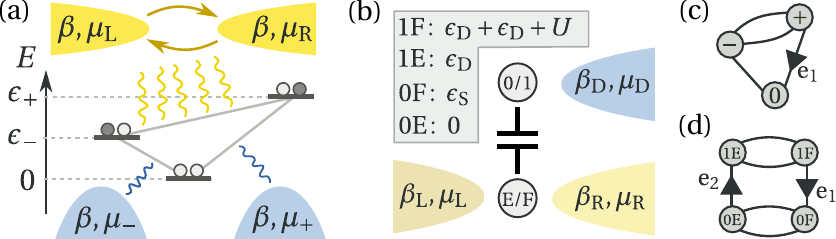}
    \caption{(a) A tunneling current through a quantum point contact is used to measure the occupancy of a double quantum dot coupled to two leads at chemical potentials $\mu_\pm$.
    (b) Two capacitatively coupled quantum dots constitute a physical implementation of Maxwell's demon, \cf Ref.~\cite{Strasberg.etal2013}.
      The lower dot (the system) is coupled to two reservoirs L and R at different chemical potentials $\beta_\mathrm{L},\mu_\mathrm{L}$ and $\beta_\mathrm{R},\mu_\mathrm{R}$.
      The upper one (the demon) interacts only with a single reservoir at inverse temperature $\beta_\text{D}$ and chemical potential $\mu_\text{D}$.
      If both dots are occupied, we have an additional contribution $U$ to the single-electron energy levels $\eps_\text{S}$ and $\eps_\text{D}$ .
    (c,d) Abstract network of states (a,b) with current supporting edges.
    }
    \label{fig:nano-examples}
    \vspace{-1.3em}
  \end{figure}

  \paragraph{Applications and examples ---}
  As is the case for equilibrium systems, our result can be used to infer response relations from fluctuations and vice versa.
  This is useful if one of these observables is more easily accessible in a given setup:
  in experiments it may be easy to tune external forces, while measuring detailed counting statistics provides a bigger challenge.
  Conversely, changing physical parameters in computer simulations may come at the expense of computation time for new simulations, while fluctuations are easily accessible within a single set of parameters.

  In the following we illustrate the application and consequences of our result for two real systems far from equilibrium.
  Both of them are nanoscopic devices, which have been studied theoretically \cite{Cuetara.Esposito2015,Strasberg.etal2013} and implemented experimentally \cite{Fujisawa.etal2006,Koski.etal2015}.
  In the first one, Fig.~\ref{fig:nano-examples}(a), a tunneling current created by a voltage difference $V=\mu_\mathrm{R} - \mu_\mathrm{L}$ through a quantum point contact (QPC) is used to measure the full counting statistics of a double quantum dot (DQD), which itself is coupled to two reservoirs at chemical potentials $\mu_\pm$~\cite{Fujisawa.etal2006,Cuetara.Esposito2015}.
  The setup is placed in a cryostat at temperature $\beta$.
  The DQD can be empty ($E=0$) or in one of two electronic states with energies $E = \eps_\pm$.
  Electrons tunneling through the QPC induce transitions in the DQD, which can be effectively reduced to two physically distinguishable mechanisms represented by the two upper edges in Fig.~\ref{fig:nano-examples}(c), \cf Ref.~\cite{Cuetara.Esposito2015} for the details.
  The chemical potential $\mu_{+}$ serves as a local coupling parameter for the current along edge $\edge_1$, which has motance $\mot_{\edge_1} = \beta(\eps_+ - \mu_+)$.
  Its conjugate current $\jmath_\alpha$ has weight $\inc^\alpha_{\edge_1} = \partial_{\mu_+}\mot_{\edge_1}= -\beta$.
  Perturbations around the stalling value  $\mu_+^\star$ obey the nonequilibrium FDR, Eq.~\eqref{eq:stallingFDR}.
  In terms of the edge current, whose statistics are experimentally accessible~\cite{Fujisawa.etal2006}, the FDR reads $\partial_{\mu_+} J_{\edge_1}^\star = -\beta \diffu^\star_{\edge_1,\edge_1}$.
  It can be used to \emph{measure} the temperature $\beta$ in situations far from equilibrium, when it is not clear if the dissipation due to the measuring current heats the device above the assumed cryostat temperature.
  Our second example, Fig.~\ref{fig:nano-examples}(b), is a capacitively interacting double quantum dot, which has recently been used as a physical implementation of Maxwell's demon~\cite{Strasberg.etal2013,Koski.etal2015}. 
  The upper dot, which constitutes the demon, is coupled to a reservoir with inverse temperature $\beta_\text{D}$ and chemical potential $\mu_\text{D}$.
  Similar to the first example, the lower dot couples to two reservoirs L and R, with a chemical potential difference $V = \mu_\mathrm{L}-\mu_\mathrm{R}$.
  We are interested in the fluctuating current $\jmath_\alpha = (\eps_\text{D} + U - \mu_\text{D})\jmath_{\edge_1} - (\eps_\text{D} - \mu_\text{D})\jmath_{\edge_2}$, associated with the transport of energy into the upper (blue) reservoir, \cf Fig.~\ref{fig:nano-examples}(b).
  Notice that in this case the physical current is supported on \emph{multiple} edges, \cf Fig.~\ref{fig:nano-examples}(d).
  The structure of the transitions ensures internal stalling:
  Due to conservation of probability, a vanishing stationary current along edge $\edge_1$ implies a vanishing current along edge $\edge_2$ and vice versa.
  The motances of these edges are $\mot_{\edge_1} = \beta_\text{D}(\eps_\text{D} + U - \mu_\text{D})$  and $\mot_{\edge_2} = -\beta_\text{D}(\eps_\text{D}  - \mu_\text{D})$, which implies that $\jmath_\alpha$ is conjugate to the inverse temperature $\beta_\text{D}$.
  Thus, at stalling the nonequilibrium FDR~\eqref{eq:stallingFDR} holds, which reads in terms of the edge currents as
  \begin{align}
    U \partial_{\beta_\text{D}} J_{\edge_1}^\star =& (\eps_\text{D}+U-\mu_\text{D})^2 \diffu^\star_{\edge_1,\edge_1} + (\eps_\text{D}-\mu_\text{D})^2 \diffu^\star_{\edge_2,\edge_2}\nonumber \\
    &- 2(\eps_\text{D}+U+\mu_\text{D})(\eps_\text{D}-\mu_\text{D})\diffu^\star_{\edge_1,\edge_2}. 
    \label{eq:demonFDR}
  \end{align}
  As an illustration of our result, in Fig.~\ref{fig:FDR-DQD} we show the difference $\partial_{\mu_\text{D}}J_{\alpha} - \diffu_{\alpha,\alpha}$ between conjugate response and fluctuations of the energy current into the demon bath as a function of demon temperature $\beta_\text{D}$ and lead voltage $V=\mu_\text{L} - \mu_\text{R}$.
  Notice that we work with reduced units where $\beta_\text{D}$ and $V$ are measured in units of the (inverse) thermal energy $\beta$.
  The black and yellow solid lines indicate the stalling values for energy  and charge current, respectively.
  Due to conservation of energy and matter, these are the only two physical currents flowing through the system~\cite{Cuetara.Esposito2015}.
  The two stalling lines cross when the system is at equilibrium, \ie when the temperature of the whole system is uniform $\beta_\text{D}=1$ and the voltage difference between the two leads vanishes, $V=0$.
  The orange dashed line indicates parameter values where the FDR~\eqref{eq:stallingFDR} holds.
  \begin{figure}[t]%
    \centering
    \includegraphics[width=.48\textwidth]{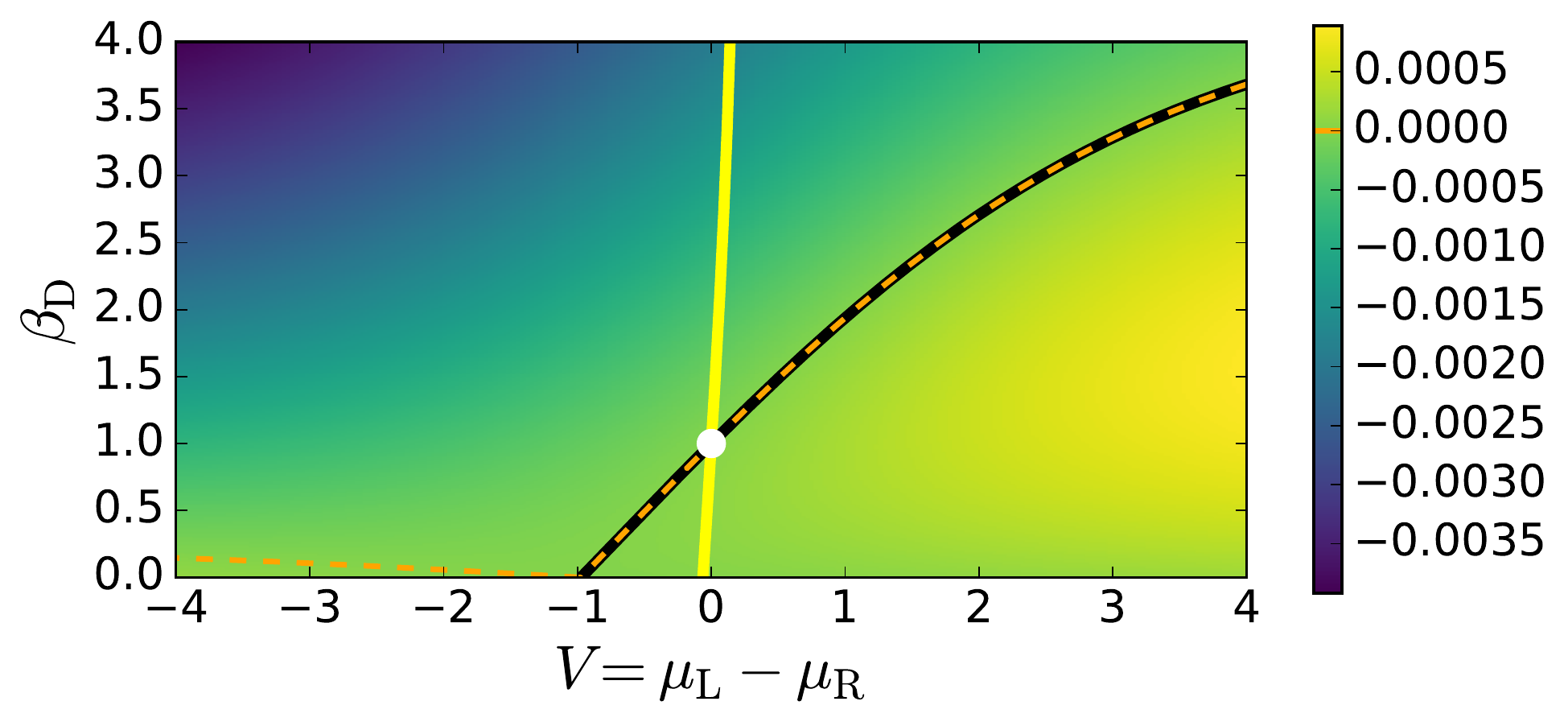}\vspace{-1em}
    \caption{Difference $\partial_{\beta_{\mathrm{D}}}J_\alpha - \diffu_{\alpha,\alpha}$ for the energy current to the demon bath.
    On the black stalling line $\beta_\text{D}^\star(V)$ the FDR holds (dashed orange line).
    The yellow line is the stalling line for the charge current between.
    At equilibrium, $(\beta_\text{D}, V) = (0,1)$, both currents vanish.
    The full set of model parameters with their numerical values is described in the supplementary material.
    }%
    \label{fig:FDR-DQD}%
    \vspace{-1em}
  \end{figure}%

  \paragraph{Proof of our main result --- }
  We first focus on several currents supported on individual edges and then move to physical currents;
  full details are deferred to the Supplementary Material.
  Recent work on the theory of large deviations \cite{Lebowitz.Spohn1999,Andrieux.Gaspard2007,Harris.Schuetz2007,BulnesCuetara.etal2011,Wachtel.etal2015,Touchette2009} has produced analytical methods to access asymptotic current statistics for Markov jump processes.
  Letting $\jmath_{\edge_\mu}$ be the currents that flow along a subset of edges labeled by $\mu$, the central quantity for calculations is the so-called tilted generator $\tmat(\vec{q}; \xs)$.
  It is obtained from the generator of the Markov jump process by replacing the off-diagonal entries corresponding to transitions $\w_{\pm\edge_\mu} \to \w_{\pm\edge_\mu} e^{\pm q_\mu}$, where $\vec{q} = (q_\mu)_\mu$ are auxiliary counting variables and the dependency on $\xs$ is inherited from the rates.
  The largest eigenvalue $\lambda(\vec{q};\xs)$ of $\tmat(\vec{q};\xs)$  is the scaled cumulant generating function (SCGF), whose first and second derivatives produce the averaged edge currents and the edge diffusivities
  \begin{align}
    J_{\edge_\mu} = \left.\frac{\partial \lambda}{\partial q_\mu}\right\vert_{\vec{q}=0} \text{ and }\,\, \diffu_{\edge_\mu,\edge_\nu} = \frac{1}{2}\left.\frac{\partial^2 \lambda}{\partial q_\mu \partial q_\nu}\right\vert_{\vec{q}=0}.
    \label{eq:cumulants}
  \end{align}
  Our first mathematical result concerns the determinant $\Delta(\vec{q};\xs):=\det\tmat(\vec{q};\xs)$.
  In particular, the crucial {\bf Theorem 1} states when only rates $\w_{\pm\edge_\mu}$ depend on the corresponding parameter $x_\mu$, in such a way that the motance increases linearly as  $\partial_{x_\mu} B_{\edge_\mu} = 1$, then there exist constants $\xs^\star$ for which
  \begin{align}
      \Delta(\xs^\star - \xs; \xs) = 0,
      \label{eq:determinant-jarzynski}
    \end{align}
  independent of $\xs$.
  As corollaries, at $\xs = \xs^*$ the average edge currents vanish, $J_{\edge_\mu}(\xs^\star) = 0$, therefore $\xs^\star$ can be identified as the stalling values of the parameters  ({\bf Corollary 1}).
  Furthermore ({\bf Corollary 2}), at stalling the mixed FDR holds
    \begin{align}
      \partial_{\x_\nu} J_{\edge_\mu}(\xs^\star) +  \partial_{\x_\mu} J_{\edge_\nu}(\xs^\star) = 2 D_{\edge_\mu,\edge_\nu}(\xs^\star). 
      \label{eq:edge-stalling-response}
    \end{align}
  Both corollaries can be obtained by taking mixed total derivatives of Eq.~\eqref{eq:determinant-jarzynski}, evaluating at stalling $\xs=\xs^\star$ and using the fact that the SCGF $\lambda(\vec{q};\xs)$ is the isolated dominant eigenvalue of $\tmat(\vec{q};\xs)$, and thus a factor of $\Delta(\vec{q};\xs)$.
  Finally, in our main {\bf Theorem 2} we consider a physical current
    \begin{align}
      \jmath_\alpha = \sum_{\edge_\mu \in \edges_\alpha} \inc^\alpha_{\edge_\mu} j_{\edge_\mu},
    \end{align}
  supported on several edges $e_\mu$, such that its conjugate force $h_\alpha$ is local, in the sense that $\partial_{h_\alpha} \w_{\pm\edge} = 0$ for all $\edge \not \in \edges_\alpha$.
  Under the assumption of LDB and internal stalling we prove the main result
    \begin{align}
      \partial_{h_\alpha}J_\alpha(\xs^\star) = \diffu_{\alpha,\alpha}(\xs^\star).
      \label{eq:main-result-pheno}
    \end{align}
  Assuming the former results, the proof of the latter is straightforward.
  By multi-linearity of the cumulants, we rewrite the diffusivity $\diffu_{\alpha,\alpha}$ as a linear combination of edge diffusivities $\diffu_{\edge_\mu,\edge_\nu}$.
  Using Eq.~\eqref{eq:edge-stalling-response} we express the latter in terms of edge responses.
  LDB prescribes $x_\mu = x_\mu(h^\alpha) := \inc^\alpha_{\edge_\mu}h_\alpha$ and after substituting $\inc^\alpha_{\edge_\mu}=\partial_{h_\alpha} x_\mu$, we use the chain rule and find the derivative with respect to $h_\alpha$.
  We provide the details of the proof in the supplementary material and note that it in fact provides a more general mixed FDR analogous to Eq.~\eqref{eq:edge-stalling-response} for the physical currents.


  The above outline of the proof allows to appreciate the conditions on which it stands.
  It further sheds light on the occurrence of nonequilibrium FDRs in other models which do not obey the requirements of local coupling and uniform stalling.
  From the logic of the proof we realize that the latter conditions are sufficient to ensure the validity of the determinant relation~Eq.~\eqref{eq:determinant-jarzynski}.
  However, this does not mean that they are necessary.
  In the case of the  models described in Refs.~\cite{Lau.etal2007} and~\cite{Liepelt.Lipowsky2007}, which treat the dynamics of the molecular motor kinesin under the influence of chemical potentials and mechanical forces, an FDR relates the response of the motor velocity around stalling to an applied mechanical force~\cite{Altaner.etal2015}.
  While these models are thermodynamically consistent in the sense that they satisfy LDB, the first model does not satisfy uniform stalling whereas in the latter model perturbations are not local.
  In both cases, the occurrence of a nonequilibrium FDR at stalling is the consequence of a determinant relation analogous to Eq.~\eqref{eq:determinant-jarzynski}.
  For Ref.~\cite{Lau.etal2007}, it follows from a global symmetry for the entire tilted generator, whereas for Ref.~\cite{Liepelt.Lipowsky2007} it relies on symmetries which are assumed for experimentally inaccessible transition rates.
  Importantly, in these two examples the stalling FDR is not robust against changing kinetic parameters such as activation barriers --- which means that either the modeling or the physics of these systems are fine-tuned.
  Such a fine-tuning may have important concequences for other properties of the system, \eg (fluctuations of) the efficiency of conversion processes \cite{Brandner.etal2013}. 
 Investigating these ideas further goes beyond the scope of this Letter. 
  However, they underline the importance of identifying general conditions for the validity of FDRs in systems far from equilibrium.
  
  
  \paragraph{Conclusion ---}
  In this work we identified local coupling and internal stalling as sufficient conditions for the validity of an equilibrium-like FDR in systems far from equilibrium. 
  Stalling conditions are commonly implemented in nanoscopic electronic devices and have been used to probe physical properties of small biological systems like molecular motors.
  The main open question from a theoretical standpoint is finding the general conditions such that stalled currents exhibit internal equilibration.
  While internal equilibration of all transitions is the definition of equilibrium, our stalling FDR is valid in situations where only some of these transitions are equilibrated while other, unobserved currents of arbitrary magnitude are present.
%

  As a final remark let us return to the questions formulated in the introduction.
  Our result states that although a system may at first glance resembles an equilibrium system in the sense that it obeys an FDR, it does not need to be so.
  The validity of an FDR is not a sufficient hallmark for equilibrium conditions.
  A true distinction between equilibrium and nonequilibrium conditions requires more scrutiny.
  For example, checking the value of the third moment (skewness) constitutes a more thorough, second glance.

  \paragraph*{Acknowledgments.}
  This research was supported by the National Research Fund Luxembourg in the frame of the project FNR/A11/02 and of the Postdoc Grant 5856127.
  Additional funding was granted by the European Research Council (project 681456).

  \bibliography{response}

\onecolumngrid
\appendix
\include{supp-input-arxiv}

\end{document}

%% file: responsemacros.tex
\newcommand{\etc}{\textit{etc}}

\newcommand{\eg}{e.g., }
\newcommand{\ie}{i.e., }
\newcommand{\cf}{cf.\ }

\newcommand{\edges}{\ensuremath{\mathcal{E}}} 

\newcommand{\edge}{\ensuremath{{\boldsymbol{\mathrm{e}}}}} 

\newcommand{\w}{\ensuremath{w}}

\newcommand{\tmat}{\ensuremath{\mathbb W}} 

\newcommand{\mot}{\ensuremath{B}} 

\newcommand*\df{\mathop{}\!\mathrm{d}}

\renewcommand{\bar}{\ensuremath{\overline}} 

\newcommand{\audretsch}[1]{\ldots} 

\renewcommand{\vec}[1]{\ensuremath{\mathbf{\boldsymbol{#1}}}} 
\newcommand{\bvec}[1]{\ensuremath{\mathbf{\boldsymbol{#1}}}} 

\newcommand{\eps}{\epsilon} 
\renewcommand{\epsilon}{\varepsilon} 

\newcommand{\kb}{\ensuremath{k_\mathrm{B}}}

\newcommand{\diffu}{\ensuremath{D}}
\newcommand{\inc}{\ensuremath{d}}

\newcommand{\x}{\ensuremath{x}}

\newcommand{\xs}{\ensuremath{\bvec{x}}}

%% file: supp-input-arxiv.tex
\section{Supplementary Material}

\newcommand{\bs}[1] {\boldsymbol{#1}}

\newcommand{\tar}{\tau}
\newcommand{\sou}{\sigma}

\newcommand{\R}{\textup{\footnotesize R}}

\subsection{Proof of the main results}

Let us recall the setup.
We consider a network (graph) composed of $N$ vertices (states of the system) and of edges $\edge \in \edges$ connecting them (transitions).
Each edge is assigned an arbitrary orientation and we assume that the graph is connected.
We consider a continuous-time Markov jump process on such a network, with probability rates per unit time $\w_{\pm\edge}$  of performing a transition along the positive ($+\edge$) and negative ($-\edge$) direction of edge $\edge$.
The generator $\tmat$ of the Markov jump process is an $N\times N$ matrix where off-diagonal entries $w_{lm}$ amount to the total rates of jumping from  vertex $m$ to vertex $l$.
Diagonal entries the total exit rate out of a state $n$, $w_{nn} = \sum^{(n)}_{\pm\edge}w_{\pm\edge}$, where the sum runs over all directed edges leaving vertex $n$.
The tilted generator for \emph{all} edge currents, $\tmat(\bs{q})$, is obtained by replacing  $w_{\pm \edge} \to w_{\pm \edge} e^{\pm q_\edge}$ in the off-diagonal entries of the generator.
The auxiliary variables $q_\edge$ are called the counting fields.
The SCGF of the currents along a subset of edges of interest, $\edges' = \{\edge_\mu\}_\mu\subset \edges$, is the unique dominant Perron eigenvalue of the tilted generator evaluated at $q_{\edge} = 0$ for all $\edge \notin \edges'$.

\setcounter{theorem}{0}
\begin{theorem}
Let $\edges = \{\edge_\mu\}_\mu$ be a subset of the edge space of the network, and consider a parametrization of the rates $w \to w(\bs{\xs})$, $\xs=\{\x_\mu\}_\mu$, such that $\partial_{\x_\mu}\w_{\pm\edge_\nu} = 0$ for all $\mu \neq \nu$. We further suppose that
\begin{align}
\partial_{x_\mu} \ln \frac{\w_{+\edge_\mu}}{\w_{-\edge_\mu}} = 1. \label{eq:motancemotance}
\end{align}
Then, there exist constants $\xs^\star = (\x_\mu^\star)_\mu$, such that the  determinant of the tilted generator for the currents along edges $\edge_\mu \in \edges$ obeys
  \begin{align}
    \Delta(\xs^\star - \xs; \xs) = 0,\quad \forall \xs.
    \label{eq:determinant-jarzynski-app}
  \end{align}
\end{theorem}

\begin{proof}

We first prove the result under the assumptions that edges $\edge_\mu$ do not share a vertex, and that the network does not allow for multiple edges between the same vertices. These conditions do not have any structural role, but including them in a general proof would make the notation awkward. We then prove that the theorem holds in all generality by showing that adding one further arbitrary edge, be it in parallel or coincident with some other edge, does not modify our argument.  Then, an edge is uniquely identified by its end vertices. If $+\edge = n \gets m$  we define $w_{nm} =  w_\edge$.

Let $\sou_\mu$ denote the source and $\tau_\mu$ the target of edge $+\edge_\mu = \tau_\mu \gets \sou_\mu$. We remind that the tilted generator has entries
\begin{align}
\tmat(\bs{q})_{l,m} = \left\{
\begin{array}{ll}
- w_{ll}, & \mathrm{if}~ l = m\\
w_{+\edge_\mu}e^{+ q_\mu}, & \mathrm{if}~ \exists \mu ~\mathrm{s.t.}~ l = \tar_\mu, m = \sou_\mu \\
w_{-\edge_\mu} e^{- q_\mu}, & \mathrm{if}~  \exists \mu ~\mathrm{s.t.}~ l = \sou_\mu, m = \tar_\mu\\
w_{lm}, & \mathrm{elsewhere}
\end{array}
 \right. .
\end{align}
From now on we omit to specify``$\exists \mu$ s.t.''. By Eq.\,(\ref{eq:motancemotance}) there exist constants $\bar{x}_\mu$ such that the motance reads
\begin{align}
\log \frac{w_{+\edge_\mu}}{w_{-\edge_\mu}} = x_\mu - \bar{x}_\mu .
\end{align}
Let us consider the Markovian generator $\overline{\tmat}$ obtained by setting to zero all rates $w_{\pm\edge_\mu} \stackrel{!}{=} 0$ corresponding to the edges $\edge_\mu$, effectively removing all edges $\edge_\mu$ from the network. Notice then that the diagonal entries (exit rates) of $\overline{\tmat}$ are given by
\begin{equation}
\begin{aligned}
\overline{w}_{ll} & = \sum_{m'} w_{m'l} , & & \mathrm{if} \; l \neq \sou_\mu, \tar_\mu  \\
\overline{w}_{ll} & = \sum_{m'  \neq \tar_\mu} w_{m'l}, & & \mathrm{if} \; l = \sou_\mu \\
\overline{w}_{ll} & = \sum_{m' \neq \sou_\mu} w_{m'l}, & &  \mathrm{if} \; l = \tar_\mu,
\end{aligned}
\end{equation}
where in this and similar expressions $\mu$ is intended to span from $1$ to $M$. Let
\begin{align}
x^\star_\mu = \bar{x}_\mu + \log \frac{V_{\tar_\mu}}{V_{\!\sou_\mu}} \label{eq:xast}
\end{align}
where $\vec{V}= (V_m)_{m=1}^M$ is some vector with nonvanishing entries. We write down the tilted generator evaluated at $\bs{q} = \bs{x}^\star - \bs{x}$:
\begin{align}
\tmat(\bs{x}^\star - \bs{x})_{l,m} & = \left\{
\begin{array}{ll}
- w_{ll}, & \mathrm{if}~ l = m\\
w_{+\edge_\mu} e^{ -x_\mu  + x^\star_\mu}, & \mathrm{if}~ l = \tar_\mu, m = \sou_\mu \\
w_{-\edge_\mu} e^{x_\mu - x^\star_\mu}, & \mathrm{if}~ l = \sou_\mu, m = \tar_\mu\\
w_{lm}, & \mathrm{elsewhere}
\end{array}
 \right. = \left\{
\begin{array}{ll}
- w_{ll}, & \mathrm{if}~ l = m\\
w_{-\edge_\mu} \frac{V_{\tar_\mu}}{V_{\sou_\mu}} , & \mathrm{if}~ l = \tar_\mu, m = \sou_\mu \\
w_{+\edge_\mu}  \frac{V_{\sou_\mu}}{V_{\tar_\mu}}, & \mathrm{if}~ l = \sou_\mu, m = \tar_\mu\\
w_{lm}, & \mathrm{elsewhere}
\end{array}
 \right. .
\end{align}
We notice that the difference between matrix $\tmat(\bs{x}^\star - \bs{x})$ and matrix $ \overline{\tmat}$ has entries
\begin{align}
[\tmat(\bs{x}^\star - \bs{x}) - \overline{\tmat}]_{l,m} = \left\{
\begin{array}{ll}
0, & \mathrm{if}~ l = m \notin \{\sou_\mu,\tar_\mu \} \\
 -  w_{+\edge_\mu}, & \mathrm{if}~ l = m  = \sou_\mu  \\
 -  w_{-\edge_\mu} , & \mathrm{if}~ l = m  = \tar_\mu  \\
w_{-\edge_\mu} \frac{V_{\tar_\mu}}{V_{\sou_\mu}}  & \mathrm{if}~ l = \tar_\mu, m = \sou_\mu \\
w_{+\edge_\mu}  \frac{V_{\sou_\mu}}{V_{\tar_\mu}}, & \mathrm{if}~ l = \sou_\mu, m = \tar_\mu\\
0, & \mathrm{elsewhere}
\end{array}
 \right. .
\end{align}
We now evaluate this matrix on vector $\vec{V} $ obtaining
\begin{align}
\left[ \tmat(\bs{x}^\star - \bs{x}) - \overline{\tmat} \right] \vec{V} = \sum_\mu \Big(  -  w_{+\edge_\mu}  V_{\sou_\mu}
-  w_{-\edge_\mu} V_{\tar_\mu} + w_{-\edge_\mu} \frac{V_{\tar_\mu}}{V_{\sou_\mu}} V_{\sou_\mu} + w_{+\edge_\mu}  \frac{V_{\sou_\mu}}{V_{\tar_\mu}} V_{\tar_\mu}  
\Big) = 0. \label{eq:zero}
\end{align}

Notice that this identity holds independently of the vector $\vec{V}$. Now, we let $\vec{V}$ be an arbitrary null eigenvector of $\overline{W}$ with all non-negative entries, $\overline{\tmat} \vec{V} = 0$. Then the last expressions tells us that $\vec{V}$ is also a null eigenvector of $\tmat(\bs{x}^\star - \bs{x})$ and we conclude.

Let us now consider an additional edge $\edge_{M+1}$ and an additional parameter that is consistent with the hypothesis of the theorem, in particular Eq.\,(\ref{eq:motancemotance}). 
Notice that edge $\edge_{M+1}$ might be in parallel with any other edge in the network, or it might have some vertex in common with one of the edges $\edge_\mu$, thus covering both situations excluded above. Without loss of generality, we can rearrange the ordering of vertices in such a way that edge $\edge_{M+1} = 1 \gets 2$ connects the first two vertices.  Let its tilted generator be denoted by $\tmat^{(M+1)}(\bs{q},q_{M+1})$, where $\tmat^{(M)}(\bs{q}) = \tmat(\bs{q})$. Let us also consider the generator of the dynamics $\overline{\tmat}^{(M+1)}$ where all edges $\edge_\mu$ and edge $\edge_{M+1}$ are removed. We have
\begin{align}
\tmat^{(M+1)} - \tmat^{(M)} & = \left( \begin{array}{ccc} 0 & w_{+\edge_{M+1}}  \left( e^{q_{M+1} }-1 \right) \\  w_{-\edge_{M+1}}  \left( e^{-q_{M+1} }-1 \right) & 0 & \ddots \\
&  \ddots &  \ddots  \end{array} \right)   \\
\overline{\tmat}^{(M+1)} - \overline{\tmat}^{(M)} & =   \left( \begin{array}{ccc} w_{-\edge_{M+1}}  & - w_{+\edge_{M+1}}  \\  - w_{-\edge_{M+1}} & w_{+\edge_{M+1}} &  \ddots \\
 & \ddots & \ddots  \end{array} \right),
\end{align}
where all entries outside the principal (upper left) $2\times2$ block are zero. We proceed like above, denoting $(\bs{y}^\star,y_{M+1})$ the new set of parameters obtained from Eq.\,(\ref{eq:xast}) by replacing $\vec{V}$ with $\vec{U}$:
\begin{subequations}
\begin{align}
y^\star_\mu & = \bar{x}_\mu + \log \frac{U_{\tar_\mu}}{U_{\!\sou_\mu}}, \quad \mu = 1,\ldots, M \\
y^\star_{M+1} & = \bar{y}_\mu + \log \frac{U_{1}}{U_{2}} 
\end{align}
\end{subequations}
We arrive at
\begin{multline}
 \tmat^{(M+1)}(\bs{y}^\star-\bs{y},y^\star_{M+1}-y_{M+1}) - \overline{\tmat}^{(M+1)} 
=  \tmat^{(M)}(\bs{y}^\star-\bs{y})  - \overline{\tmat}^{(M)} + \left( \begin{array}{ccc} -w_{-\edge_{M+1}}  & w_{-\edge_{M+1}} 
\frac{U_{1}}{U_{2}}  
 \\  w_{+\edge_{M+1}}   \frac{U_{2}}{U_{1}} 
& -  w_{+\edge_{M+1}} & \ddots \\
& \ddots  & \ddots  \end{array} \right).
\end{multline}
We now evaluate on $\vec{U}$. As explained above, Eq.~\eqref{eq:zero} holds irrespective of $\vec{V}$, hence it also holds on $\vec{U}$. We then obtain
\begin{multline}
\left[ \tmat^{(M+1)}(\bs{y}^\star-\bs{y},y^\star_{M+1}-y_{M+1}) - \overline{\tmat}^{(M+1)}   \right] \vec{U} 
=\left( \begin{array}{ccc} - w_{-\edge_{M+1}}    &w_{-\edge_{M+1}} 
\frac{U_{1}}{U_{2}} 
\\  w_{+\edge_{M+1}}   \frac{U_{2}}{U_{1}} 
& -  w_{+\edge_{M+1}}  & \ddots \\
& \ddots & \ddots  \end{array} \right) \vec{U}  = 0.
\end{multline}
Again, we conclude by assuming that $\vec{U}$ is a solution of $\overline{W}^{(M+1)} \vec{U} = 0$.
\end{proof}

\begin{corr}
  \label{theo:stalling-response-app}
At $\xs = \xs^\star$, the steady state currents $ J_{\edge_\mu}(\xs^\star)$ vanish.
\end{corr}
\begin{proof}
Taking the total derivative of the determinantal equation Eq.\,(\ref{eq:determinant-jarzynski-app}) with respect to parameter $x_\mu$ and evaluating at $\bs{x} = \bs{x}^\star$ we obtain
\begin{align}
\frac{\partial}{\partial q_\mu} \Delta(\bs{0};\bs{x}^\star) = \frac{\partial}{\partial x_\mu} \Delta(\bs{0};\bs{x}^\star).
\end{align}
Notice that the right-hand side vanishes identically, because
\begin{align}
  \Delta(0;\bs{x}) = 0, \quad\forall \bs{x}
  \label{eq:detzero}
\end{align}
is the determinant of a stochastic matrix.
We now use the fact that the determinant of the tilted generator is the product of the eigenvalues, $\Delta(\bs{q};\bs{x}) = \prod_i \lambda_i(\bs{q};\bs{x})$ and that the dominant eigenvalue $\lambda_1(\bs{q}) = \lambda(\bs{q})$ is the SCGF. We then obtain
\begin{align}
0 = \frac{\partial \lambda}{\partial q_\mu} (\bs{0};\bs{x}^\star) \prod_{i > 1} \lambda_i(\bs{0};\bs{x}^\star)
\end{align}
where we used the fact that $\lambda(\bs{0};\bs{x}^\star) = 0$ so that the partial derivative only affects the dominant eigenvalue.
By the Perron--Frobenius theorem, this eigenvalue is isolated such that $\prod_{i > 1} \lambda_i(\bs{0};\bs{x}^\star)\neq0$.
We then conclude by recognizing that the current is  $ J_{\edge_\mu}(\xs^\star) = \partial_{q_\mu}  \lambda(\bs{0};\xs^\star)$.
\end{proof}

\begin{corr}
At $\xs = \xs^\star$, the fluctuation-response relation holds:
  \begin{align}
    \partial_{\x_\nu} J_{\edge_\mu}(\xs^\star) +  \partial_{\x_\mu} J_{\edge_\nu}(\xs^\star) = 2 D_{\edge_\mu,\edge_\nu}(\xs^\star), \quad\forall \mu,\nu.
    \label{eq:edge-stalling-response-app}
  \end{align}
\end{corr}

\begin{proof}
The proof proceeds like the one above. Taking the second mixed total derivatives with respect to $x_\mu$ and $x_\mu'$ we obtain
\begin{align}
  0 & = \frac{d^2 }{d x_\mu d x_{\mu'} } \Delta(\bs{x}^\star - \bs{x};\bs{x}) \\
  & = \left[ \frac{\partial^2  \Delta}{\partial q_\mu \partial q_{\mu'} }
  -  \frac{\partial^2  \Delta}{\partial q_\mu \partial x_{\mu'} } 
  -  \frac{\partial^2  \Delta}{\partial x_\mu \partial q_{\mu'} } 
+  \frac{\partial^2  \Delta}{\partial x_\mu \partial x_{\mu'} } \right](\bs{x}^\star - \bs{x};\bs{x}). \label{eq:term}
\end{align}
Letting $z$ be either $x$ or $q$, we have 
\begin{align}
\frac{\partial^2  \Delta}{\partial z_\mu \partial z_{\mu'}} = 
\sum_i \frac{\partial^2  \lambda_i}{\partial z_\mu  \partial z_{\mu'}} \prod_{j\neq i} \lambda_j 
+  \sum_{\substack{i,j\\ i \neq j}} \frac{\partial  \lambda_i }{\partial z_\mu  }\frac{\partial  \lambda_i }{\partial z_{\mu'}  } \prod_{k \neq i,j} \lambda_k. 
\label{eq:secondpartial}
\end{align}
We now evaluate at $\xs = \xs^\star$. By Eq.~\eqref{eq:detzero} we realize that the last term in Eq.\,(\ref{eq:term}) vanishes. For all other terms $(z_\mu,z_{\mu'}) = (q_\mu,q_{\mu'}), (q_\mu,x_{\mu'}), (x_\mu,q_{\mu'})$, Eq.~\eqref{eq:secondpartial} yields
\begin{align}
\frac{\partial^2  \Delta}{\partial z_\mu \partial z_{\mu'}}(\bs{0};\bs{x}^\star) = 
\frac{\partial^2  \lambda}{\partial z_\mu  \partial z_{\mu'}} \prod_{j > 1} \lambda_j (\bs{0};\bs{x}^\star)
\end{align}
where we used that at stalling,  $\bs{x} = \bs{x}^\star$, both the dominant eigenvalue $\lambda(0;\xs)$ and the current $\partial_{q_\mu}\lambda(0;\xs)$ vanish.
It then follows from Eq.\,(\ref{eq:term})
\begin{align}
\left[ \frac{\partial^2  \lambda}{\partial q_\mu  \partial q_{\mu'}} - \frac{\partial^2  \lambda}{\partial q_\mu  \partial x_{\mu'}} - \frac{\partial^2  \lambda}{\partial x_\mu  \partial q_{\mu'}} \right] (\bs{0};\bs{x}^\star) = 0
\end{align}
which is the fluctuation-response relation Eq.~\eqref{eq:edge-stalling-response-app}.

\end{proof}

\begin{theorem}
  \label{theo:main-result}
  Consider a phenomenological current 
  \begin{align*}
    \jmath_\alpha = \sum_{\edge_\mu \in \edges_\alpha} \inc^\alpha_{\edge_\mu} \jmath_{\edge_\mu}, 
  \end{align*}
and a parametrization $x_\mu = x_\mu(h^\alpha) := \inc^\alpha_{\edge_\mu}h_\alpha$ consistent with the condition of detailed balance. Then at values of $h^\alpha$ where all of the edge currents vanish, the fluctuation-dissipation relation for the phenomenological currents holds:
\begin{align}\label{eq:abc}
2\diffu_{\alpha,\alpha'}(\xs^\star) = \partial_{h_\alpha'}J_\alpha(\xs^\star) + \partial_{h_\alpha}J_{\alpha'}(\xs^\star). 
\end{align}
\end{theorem}

\begin{proof}
The cumulants are multi-linear, therefore 
\begin{align}
    \diffu_{\alpha,\alpha'}(\xs^\star) = \sum_{\mu,\nu} \inc^\alpha_{\edge_\mu}\inc^{\alpha'}_{\edge_\nu} \diffu_{\edge_\mu,\edge_\nu}(\xs^\star).
\end{align}
  Using Eq.~\eqref{eq:determinant-jarzynski-app} and replacing  $\inc^\alpha_{\edge_\mu}=\partial_{h_\alpha} x_\mu$, we find
   \begin{align*}
    \diffu_{\alpha,\alpha'} & = \sum_{\mu,\nu} \inc^\alpha_{\edge_\mu}\inc^{\alpha'}_{\edge_\nu} \diffu_{\edge_\mu,\edge_\nu}(\xs^\star) \\
    &= \frac{1}{2}\sum_{\mu,\nu} \left[\inc^{\alpha'}_{\edge_\nu} \partial_{\x_\nu}  \inc^\alpha_{\edge_\mu} J_{\edge_\mu}(\xs^\star) +   \inc^\alpha_{\edge_\mu} \partial_{\x_\mu} \inc^{\alpha'}_{\edge_\nu} J_{\edge_\nu}(\xs^\star) \right ]\\
    &=  \frac{1}{2} \sum_\nu \left[ (\partial_{h_\alpha} \x_\nu)  (\partial_{\x_\nu}  J_{\alpha'})  + (\partial_{h_{\alpha'}} \x_\nu)  (\partial_{\x_\nu}  J_{\alpha}) \right]\\
    &= \frac{1}{2}\left(\partial_{h_\alpha}J_{\alpha'} + \partial_{h_\alpha'}J_{\alpha} \right).
  \end{align*}
\end{proof}
Notice that this equation is a generalization of the mixed fluctuation-response relation Eq.~\eqref{eq:edge-stalling-response-app} to physical currents.

\subsection{Transition rates and numerical values of the double quantum dot model, Figs. 1(b,d) and 2}
For completeness, we give the transition rates for the double quantum dot  described in Figs.~1(b,d) and Fig.~2, \cf Ref.~\cite{Strasberg.etal2013} for the details.
The double quantum dot can be in four states corresponding to the presence or absence of electrons in either dot.
Charging one a dot while an electron is present in the other dot requires an additional amount of energy $U$ due to Coulomb repulsion.
The charging and discharging statistics of the double quantum dot are governed by the Fermi--Dirac statistics.

For a transition involving a reservoir $\text{Y} \in \left\{\mathrm{L},\mathrm{R},\mathrm{D}\right\}$ we define the Fermi function $f_\text{Y}(x) = (1+\exp(\beta_\text{Y} x))^{-1}$, which are evaluated at the transition energies $ x = E-\mu_{\text{Y}}$, where $E$ is the energy of the electronic system.
The kinetic factors for the charging transitions may distinguish whether the other dot is empty ($\Gamma_\text{Y}$) or occupied ($\Gamma_\text{Y}^U$).
Denoting the horizontal edges in Fig.~1(d) from top to bottom by $\edge_3,\dots\edge_6$ with an orientation that is aligned with edges $\edge_1$ and $\edge_2$ we have the charging transition rates
\begin{align*}
w_{-\edge_1} &= \Gamma^U_\mathrm{D}(f_\mathrm{D}(\eps_\mathrm{D}  + U - \mu_\mathrm{D}),
&w_{+\edge_2} &= \Gamma_\mathrm{D}(f_\mathrm{D}(\eps_\mathrm{D} - \mu_\mathrm{D})),\\
w_{+\edge_3} &= \Gamma^U_\mathrm{L}(f_\mathrm{L}(\eps_\mathrm{L} + U - \mu_\mathrm{L}),
&w_{+\edge_4} &= \Gamma^U_\mathrm{R}(f_\mathrm{R}(\eps_\mathrm{R} + U - \mu_\mathrm{R})),\\
w_{-\edge_5} &= \Gamma_\mathrm{L}(f_\mathrm{L}(\eps_\mathrm{L}  - \mu_\mathrm{L}),
&w_{-\edge_6} &= \Gamma_\mathrm{R}(f_\mathrm{R}(\eps_\mathrm{R}  - \mu_\mathrm{R})).  
\end{align*}
The corresponding reversed rates, which correspond to a discharge, have the same rates replacing $f_\text{Y}(x)\to1-f_\text{Y}(x)$.
Note that the motance, \ie the ratio of a charging and discharging transition reads 
\begin{align*}
  \ln\frac{f_\text{Y}(x)}{(1-f_\text{Y}(x))} = -\beta_\text{Y} x. 
\end{align*}
The numerical values used for Fig.~2 are: 
\begin{align*}
\Gamma_\mathrm{D} &= \Gamma^U_\mathrm{D} = \Gamma_\mathrm{R} = \Gamma^U_\mathrm{L} = 0.5, \qquad
\Gamma_\mathrm{L} = \Gamma^U_\mathrm{R} = 1.5,\\
\mu_\mathrm{D} &= 1,\quad
\mu_\mathrm{L} = 1+V/2,\quad
\mu_\mathrm{R} = 1-V/2,\quad\\
\beta_\mathrm{R} &= \beta_\mathrm{L} = \beta = 1,\quad
\eps_\mathrm{D}=\eps_\mathrm{S} = 1\quad
\text{ and }U = 0.5.
\end{align*}
The voltage difference $V$ and the ``demon'' inverse temperature $\beta_\text{D}$ are kept variable, \cf Fig.~2.